\documentclass[11pt,a4paper]{article}

\usepackage[utf8]{inputenc}
\usepackage{amsmath}
\usepackage{amsfonts}
\usepackage{amssymb}
\usepackage{amsthm}

\title{Some Results on Circuit Lower Bounds and Derandomization of Arthur-Merlin Problems}
\author{D. M. Stull\footnote{Research supported in part by National Science Foundation Grants 1247051 and 1545028.}}

\newtheorem{theorem}{Theorem}
\newtheorem{lemma}{Lemma}
\newtheorem{corollary}{Corollary}

\theoremstyle{definition}
\newtheorem*{definition}{Definition}

\newcommand{\sst}{\{0,1\}^*}
\newcommand{\sstn}{\{0,1\}^n}

\newcommand{\AM}{\mathsf{AM}}
\newcommand{\Poly}{\mathsf{P}}
\newcommand{\prAM}{\mathsf{prAM}}
\newcommand{\NPpoly}{\mathsf{NP}/poly}
\newcommand{\NP}{\mathsf{NP}}
\newcommand{\coNP}{\mathsf{coNP}}
\newcommand{\augAM}{\mathsf{AugAM}}
\newcommand{\coaugAM}{\mathsf{coAugAM}}
\newcommand{\prcoaugAM}{\mathsf{pr-coAugAM}}
\newcommand{\praugAM}{\mathsf{prAugAM}}
\newcommand{\avmpraugam}{\mathsf{prM}(\mathsf{AM}||\mathsf{coNP})}

\newcommand{\nsize}{\mathsf{NSIZE}}
\newcommand{\svsize}{\mathsf{SV}\mathsf{SIZE}}
\newcommand{\pspace}{\mathsf{PSPACE}}
\newcommand{\Ppoly}{\mathsf{P}/poly}
\newcommand{\SigmaP}{\Sigma_2\mathsf{P}}
\newcommand{\SigmaSub}{\Sigma_2\mathsf{SubEXP}}
\newcommand{\SigmaEXP}{\Sigma_2\mathsf{EXP}}
\newcommand{\SigmaE}{\Sigma_2\mathsf{E}}

\newcommand{\N}{\mathbb{N}}
\newcommand{\MA}{\mathsf{MA}}
\newcommand{\PiP}{\Pi_2\mathsf{P}}
\begin{document}

\maketitle
\begin{abstract}
We prove a downward separation for $\mathsf{\Sigma}_2$-time classes. Specifically, we prove that if $\SigmaE$ does not have polynomial size non-deterministic circuits, then $\SigmaSub$ does not have \textit{fixed} polynomial size non-deterministic circuits. To achieve this result, we use Santhanam's technique \cite{Santhanam09} on augmented Arthur-Merlin protocols defined by Aydinlio\u{g}lu and van Melkebeek \cite{AvM12}. We show that augmented Arthur-Merlin protocols with one bit of advice do not have fixed polynomial size non-deterministic circuits. We also prove a weak unconditional derandomization of a certain type of promise Arthur-Merlin protocols. Using Williams' easy hitting set technique \cite{Williams16}, we show that $\Sigma_2$-promise $\AM$ problems can be decided in $\SigmaSub$ with $n^c$ advice, for some fixed constant $c$.
\end{abstract}

\section{Introduction}
The power of non-uniform (i.e., circuit) models of computation is a central topic in theoretical computer science. In addition to being intrinsically interesting, proving circuit lower bounds for uniform classes has many important consequences. Indeed, proving that $\NP$ does not have polynomial size Boolean circuits would imply that $\Poly \neq \NP$. 

Circuit lower bounds also have strong connections with the derandomization of probabilistic complexity classes. The so called ``hardness vs. randomness" paradigm is based on the idea that if a language has high circuit complexity, we can use the language to derandomize probabilistic classes using pseudorandom generators. Babai, et al, \cite{BFNW93} used this idea for ``low-end" derandomization of $\mathsf{BPP}$. They showed that if $\mathsf{E}$ does not have polynomial size circuits, then $\mathsf{BPP}$ can be derandomized in subexponential time infinitely often. Subsequently, Impagliazzo and Wigderson \cite{IW97} gave a ``high-end" derandomization of $\mathsf{BPP}$. They proved that, if $\mathsf{E}$ does not have $2^{\epsilon n}$ size Boolean circuits, then $\mathsf{BPP} = \mathsf{P}$. We now know that, in certain settings, circuit lower bounds and derandomization are \textit{equivalent}. Impagliazzo et al, showed that any non-trivial derandomization of the class $\mathsf{MA}$ implies that $\mathsf{NEXP}$ does not have polynomial size deterministic circuits \cite{IKW02}. Kabanets and Impagliazzo subsequently proved that derandomizing the well known Polynomial Identity Testing problem is equivalent to circuit lower bounds \cite{KI04}.

Aydinlio\u{g}lu and van Melkebeek \cite{AvM12} have recently introduced \textit{augmented Arthur-Merlin protocols}  to extend the equivalence of circuit lower bounds and derandomization to \textit{non}-\textit{deterministic} circuits. Using augmented $\AM$, Aydinlio\u{g}lu and van Melkebeek showed that derandomizing promise $\AM$ in $\SigmaSub$ with $n^\epsilon$ bits of advice is equivalent to polynomial size non-deterministic circuit lower bounds for $\SigmaE$.

In this paper, we investigate non-deterministic circuit lower bounds of uniform classes. We prove that non-deterministic circuit lower bounds translate downward for $\mathsf{Sigma}_2$-time classes. Specifically, we show that if $\SigmaE$ does not have polynomial size non-deterministic circuits, then $\SigmaSub$ does not have \textit{fixed} polynomial size non-deterministic circuits. To prove this result, we give fixed polynomial size non-deterministic lower bounds for augmented Arthur-Merlin protocols, which may be of independent interest. To achieve this, we use a technique developed by Santhanam \cite{Santhanam09} to prove analogous results for $\mathsf{MA}$. 

While circuit lower bounds are notoriously hard to prove, there has been important progress in this direction. Kannan proved that $\SigmaP \cap \PiP$ does not have fixed polynomial size deterministic circuits \cite{Kannan82}. Subsequently, Kobler and Watanabe improved this lower bound holds for the weaker class $\mathsf{ZPP}^{\mathsf{NP}}$ \cite{KW98}. Cai was able to strengthen this further by showing that $\mathsf{S_2P}$ does not have fixed polynomial size deterministic circuits \cite{Cai07}. Vinodchandran proved  fixed $n^k$ circuit lower bounds for the class $\mathsf{PP}$. Santhanam, using tools from interactive proof protocols \cite{LFKN92}, \cite{Shamir92} and program checking \cite{BK95}, proved that $\mathsf{MA}$ with one bit of advice does not have fixed polynomial size deterministic circuits. There have been fewer unconditional lower bounds for \textit{non-deterministic} circuits. The smallest class known to have fixed polynomial size non-deterministic circuits is $\mathsf{S}_2\mathsf{P}^{\mathsf{NP}}$, which follows by relativizing Cai's result \cite{Cai07}. In this paper, we show that Santhanam's technique can be applied to the augmented Arthur-Merlin protocols of \cite{AvM12}. This improves the smallest class known to have fixed size non-deterministic circuits.

One of the principal interests in proving non-deterministic circuit lower bounds is the derandomization of $\AM$. The work of Klivans and van Melkebeek \cite{KvM02}; Shaltiel and Umans \cite{SU06} shows that derandomization of $\prAM$ follows from non-deterministic circuit lower bounds. Recently, progress has been made on achieving non-trivial derandomization of $\AM$ in $\Sigma_2$-time classes. Kabanets \cite{Kabanets01} using his ``easy witness" technique, and Gutfreund et al \cite{GST03}, gave unconditional derandomization of $\AM$ in pseudo-$\SigmaSub$. Williams', using his ``easy hitting set" technique, recently showed that $\AM$ is contained in $\SigmaSub$ with fixed $n^c$ advice \cite{Williams16}. In this paper, we investigate derandomization of \textit{promise} $\AM$. We use Williams easy hitting set technique to show that certain promise $\AM$ protocols can be unconditionally derandomized in $\SigmaSub$, with fixed $n^c$ bits of advice.

\section{Preliminaries}
We will assume familiarity with the complexity classes $\NP$, $\SigmaP$, $\pspace$ as well as their exponential- and subexponential-time counterparts. For a language $L$ and integer $n$, we denote the restriction of $L$ to $n$ by $L_{=n}$, consisting of all strings $x \in \sstn \cap L$. We denote the complement of a language $L$ by $\overline{L}$. For a language $L$ and a complexity class $\mathcal{C}$, we say that $L$ is infinitely often in $\mathcal{C}$, denoted $L \in$ i.o-$\mathcal{C}$, if there is a language $A \in \mathcal{C}$ such that for infinitely many $n \in \N$, $L_{=n} = A_{=n}$.

\subsection{Non-deterministic circuits}
A \textit{non-deterministic Boolean circuit} $C$ is a Boolean circuit which receives two inputs, $x$ of length $n$ and a second input $y$. We say that $C$ accepts input $x$ if there is a string $y$ such that $C(x, y) = 1$. Otherwise, we say that $C$ rejects $x$. The \textit{size} of a non-deterministic circuit is the number of its connections. For a constant $k \in \N$, the class $\nsize(n^k)$ consists of all languages $L$ for which there is a family of non-deterministic circuits $\{C_n\}_{n \in \N}$ such that $C_n$ decides $L_{=n}$ and $size(C_n) = n^k$. The class $\nsize(poly)$ is the union of $\nsize(n^k)$ over all constants $k \in \N$.

A \textit{partial single-valued nondeterministic (PSV) circuit} is a Boolean circuit $C$ which receives two inputs, $x$ of length $n$ and a second input $y$, and has two output gates, \textit{value} and \textit{flag}, so that the following holds for every $x \in \sstn$. 
\begin{enumerate}
\item For every $y_1, y_2$, if $C(x, y_1)$ and $C(x, y_2)$ have a 1 at their flag gate, then $C(x, y_1) = C(x, y_2)$.
\end{enumerate}
Circuit $C$ is a \textit{total single-valued (TSV) circuit} computing the function $f : \{0, 1\}^n\rightarrow\{0, 1\}$ if the following hold.
\begin{enumerate}
\item $C$ is a PSV circuit.
\item For every $x$, there exists some $y$ for which $C(x, y$) has 1 at its flag gate.
\end{enumerate}
For a constant $k$, the class of $n^k$ size \textit{single-valued non-deterministic circuits}, $\svsize(n^k)$, consists of all languages $L$ for which there is a family of TSV circuits $\{C_n\}_{n \in \N}$ such that $C_n$ decides $L_{=n}$ and $size(C_n) = n^k$. The class $\svsize(poly)$ is the union of $\svsize(n^k)$ over all constants $k \in \N$. Note that for any language $L$, if $L, \overline{L} \in \nsize(poly)$, then $L$ and $\overline{L}$ are in $\svsize(poly)$.

\subsection{Arthur-Merlin protocols}
Promise problems were first introduced and studied by  Even, Selman and Yacobi \cite{ESY84}. They have since been highly useful in complexity theory, and, in particular, probabilistic complexity classes. A \textit{promise problem} $\Pi = (\Pi_Y, \Pi_N)$ is a pair of disjoint sets $\Pi_Y$ and $\Pi_N$. A language $L$ \textit{agrees} with a promise problem $\Pi$ if 
\begin{enumerate}
\item $x \in L$ for every $x \in \Pi_Y$, and
\item $x \notin L$ for every $x \in \Pi_N$.
\end{enumerate}

The class of Promise Arthur-Merlin problems, $\prAM$, is the set of all promise problems $\Pi$ such that there is a polynomial time relation $R(\cdot, \cdot, \cdot)$ such that
\begin{align*}
x \in \Pi_Y  \implies & \Pr_z[(\exists y) \, R(x, y, z) = 1] \geq 2/3 \\
x \in \Pi_N  \implies & \Pr_z[(\exists y) \, R(x, y, z) = 1] \leq 1/3.
\end{align*} 
The class $\AM$ consists of the problems in $\prAM$ which are languages.

Augmented Arthur-Merlin protocols were introduced by Aydinlio\u{g}lu and van Melkebeek \cite{AvM12}. This definition is similar to $\AM$ protocols, except that there are two verifiers, Arthur and a $\coNP$ verifier $V$.  
\begin{definition}[Augmented Arthur-Merlin protocols]
The class of problems $\praugAM$\footnote{Aydinlio\u{g}lu and van Melkebeek originally denoted this class by $\avmpraugam$. We made this change for considerations of length.} consists of all promise problems $\Pi$ for which there is a constant $c$, a promise problem $\Gamma \in \prAM$ and a language $V \in \coNP$ such that
\begin{align*}
x \in \Pi_Y  \implies & (\exists y) (\langle x, y \rangle \in \Gamma_Y \land \langle x, y \rangle \in V),\\
x \in \Pi_N  \implies & (\forall y) (\langle x, y \rangle \in \Gamma_N \lor \langle x, y \rangle \notin V),
\end{align*}
where $x \in \sstn$ and $y \in \{0,1\}^{n^c}$. The class $\augAM$ consists of the problems in $\praugAM$ which are languages.
\end{definition}

\subsection{Pseudorandom Generators}
The $SAT$-\textit{relativized hardness} $H^{SAT}(G_{r, n})$ of a \textit{pseudorandom generator} $G_{r,n} : \{0, 1\}^r \rightarrow \{0, 1\}^{n}$ is defined as the minimal $s$ such that there exists an $n$-input $SAT$ oracle Boolean circuit C of size at most $s$ for which 
\begin{center}
$\vert \Pr_{x \in \{0,1\}^r}[C(G_{r, n}(x) = 1] - \Pr_{y \in \{0,1\}^n}[C(y) = 1] \vert \geq \frac{1}{s}$.
\end{center}

Klivans and van Melkebeek \cite{KvM02} showed that the pseudorandom generator constructions of \cite{BFNW93} and \cite{IW97} relativize. Specifically, they proved the following theorem.
\begin{theorem}\label{thm:nondeterPSG}
There is a polynomial-time computable function $F : \sst \times \sst \rightarrow \sst$ with the following properties. For every $\epsilon > 0$, there exist $c, d \in \N$ such that 
\begin{center}
$F : \{0,1\}^{n^c} \times \{0,1\}^{d \log n} \rightarrow \{0,1\}^n$,
\end{center}
and if $r$ is the truth table of a $c \log n$ variable Boolean function of $SAT$-oracle circuit complexity at least $n^c$, then the function $G_r(s) = F(r, s)$ is a pseudorandom generator with hardness $H^{SAT}(G_r) > n$.
\end{theorem}

Klivans and van Melkebeek showed that the existence of pseudorandom generators which are hard for $SAT$-oracle circuits derandomize $\prAM$ \cite{KvM02}.

\section{Non-deterministic Circuit Lower Bounds}
We now prove our downward separation for $\mathsf{\Sigma}_2$-time classes. We first show that $(\coaugAM \cap \augAM)/1$ does not have fixed polynomial size non-deterministic circuits. We will need the following lemma of Santhanam \cite{Santhanam09}, which builds on the ideas of Trevisan and Vadhan \cite{TV07} and Fortnow and Santhanam \cite{FS04}. 
\begin{lemma}\label{thm:SanthPSPACECompleteLanguage}
There is a $\pspace$-complete language $L$ and probabilistic polynomial-time oracle Turing machines $M$ and $M^\prime$ such that for any input x of length $n$ the following hold.
\begin{enumerate}
\item $M$ and $M^\prime$ only query their oracle on strings of length $n$.
\item If $M$ (resp. $M^\prime$) is given $L$ as its oracle and $x \in L$ (resp. $x \notin L$), then $M$ (resp. $M^\prime$) accepts with probability $1$.
\item If $x \notin L$ (resp. $x \in L$), then irrespective of the oracle, $M$ (resp. $M^\prime$) rejects with probability at least $2/3$.
\end{enumerate}
\end{lemma}

We will use the complete language of Lemma \ref{thm:SanthPSPACECompleteLanguage} to define promise Arthur-Merlin problems. Let $L$, $M$ and $M^\prime$ be as in the definition of Lemma \ref{thm:SanthPSPACECompleteLanguage}. For every PSV circuit $C$ and input $x$, let $\Pr[M^{L(C)} = 1]$ denote the probability over $M$'s random bits that $M$ accepts when given the language of $C$ as an oracle. We will also make the following assumption on the behavior of $M$ (and $M^\prime$). If $C$ is undefined at some $x^\prime$, and $M$ queries its oracle for $x^\prime$, we will assume that the oracle returns a special symbol '?' and $M$ will immediately halt and reject. Define the promise problem $\Gamma^M = (\Gamma^M_Y, \Gamma^M_N)$ by
\begin{align*}
\Gamma^M_Y &= \{\langle x, C \rangle \, | \, C \text{ is a PSV circuit s.t. } \Pr[M^{L(C)})(x) = 1] \geq 2/3\}\\
\Gamma^M_N &= \{\langle x, C \rangle \, | \, C \text{ is a PSV circuit s.t. } \Pr[M^{L(C)}(x) = 0] \geq 2/3\}.
\end{align*}

In a similar manner, and with the same assumption on the behavior of $M^\prime$, define the promise problem $\Gamma^{M^\prime} = (\Gamma^{M^\prime}_Y, \Gamma^{M^\prime}_N)$ by
\begin{align*}
\Gamma^{M^\prime}_Y &= \{\langle x, C \rangle \, | \, C \text{ is a PSV circuit s.t. } \Pr[M^{\prime L(C)})(x) = 1] \geq 2/3\}\\
\Gamma^{M^\prime}_N &= \{\langle x, C \rangle \, | \, C \text{ is a PSV circuit s.t. } \Pr[M^{\prime L(C)}(x) = 0] \geq 2/3\}.
\end{align*}

\begin{lemma}\label{lemma:isPromiseProblem}
Let $L$, $M$ and $M^\prime$ be as in the definition of Lemma \ref{thm:SanthPSPACECompleteLanguage}. Let $\Gamma^M$ and $\Gamma^{M^\prime}$ be the promise problems defined above. Then $\Gamma^M$ and $\Gamma^{M^\prime}$ are in $\prAM$.
\end{lemma}
\begin{proof}
We give the Arthur-Merlin protocol for $\Gamma^M$. The protocol for $\Gamma^{M^\prime}$ is identical. On input $\langle x, C \rangle$, the Arthur-Merlin protocol works as follows. Arthur guesses a random string $r$, and sends $r$ to Merlin. Merlin responds with a sequence of witnesses $w_1,\ldots, w_{p(n)}$. Arthur then simulates $M$ with $L(C)$ as its oracle by using the provided witnesses. That is, for every query $q_j$, Arthur simulates $C(q_j, w_j)$. If for any $j$, $C(q_j, w_j)$ does not have a $1$ at its flag gate, Arthur immediately halts and rejects. Otherwise, Arthur uses the value of $C(q_j, w_j)$ as the oracle response and continues.  

From the definition of $M$ and $M^\prime$, it is clear that $\Gamma^M$ and $\Gamma^{M^\prime}$ are in $\prAM$.
\end{proof}

The usefulness of the $\coNP$ verifier in the definition of augmented $\AM$ protocols is that it allows for us to simulate interactive proof protocols. In the deterministic circuit setting, we are able to prove that $\pspace \subseteq \Ppoly$ implies that $\pspace = \MA$. This follows from the fact that Merlin can send Arthur a Boolean circuit claiming to compute the provers strategy, and Arthur simply simulates the interactive proof protocol using this circuit as the oracle. In the non-deterministic setting, however, this method breaks down. The essential difficulty is that Arthur cannot know if the non-deterministic circuit returns ``no" on every path, or just the one Merlin gives. The inclusion of a $\coNP$ verifier allows the proof for deterministic circuits to extend to the non-deterministic setting. Using this strategy, Aydinlio\u{g}lu and van Melkebeek \cite{AvM12} proved the following Lemma.
\begin{theorem}\label{thm:AvM}
If $\pspace \subseteq \NPpoly$, then $\pspace \subseteq \augAM$.
\end{theorem}

For the sake of clarity, we will break the proof of our main theorem into two parts. The first part uses Santhanam's technique \cite{Santhanam09} to show that augmented Arthur-Merlin protocols with one bit of advice do not have fixed polynomial size \textit{SV-circuits}. We then modify this proof slightly to achieve the stronger statement, that class $(\augAM \cap \coaugAM)/1$ does not have fixed size \textit{non-deterministic circuits}.

\begin{theorem}\label{thm:fixedLowerAugAM}
For every $k \in \N$ there is a language $A \in \augAM/1$ such that $L \notin \svsize(n^k)$.
\end{theorem}
\begin{proof}
First assume that $\pspace$ has polynomial size $SV$ circuits. Then, by Theorem \ref{thm:AvM}, $\mathsf{PSPACE} \subseteq \augAM$, and the conclusion follows. So we may assume that $\pspace \nsubseteq \svsize(poly)$. 

Let $k \in \N$ and $L$ be the $\pspace$-complete language of Lemma \ref{thm:SanthPSPACECompleteLanguage}. By our assumption, $L \notin \svsize(poly)$. For every $n \in \mathbb{N}$, define the $Min(L_n) \in \mathbb{N}$ to be the size of the smallest $SV$-circuit computing $L_{=n}$. Define the language $A$ by 
\begin{center}
$A = \{x1^y \, | \, x \in L, \, 0< |x| \leq y, \; y \text{ is a power of 2 and } (y + |x|)^{k+1} \leq Min(L_n) < (2y + |x|)^{k+1}\}.$
\end{center}
We first show that $A \in \augAM/1$. Define the $\augAM$ protocol with one bit of advice as follows. On input $w$, if the advice is set to $0$, Arthur halts and rejects. If the advice bit is set to $1$, Arthur verifies that $w = x1^y$, where $0 < |x| \leq y$ and $y$ is a power of $2$. If the input is not of this form, Arthur halts and rejects. Otherwise, if the input is of the correct form, Merlin sends a non-deterministic circuit $C_L$ claiming to compute $L_{=n}$ to both verifiers. The $\coNP$ verifier $V$ checks that $C_L$ is a PSV circuit. That is, $V$ checks that for every string $x^\prime$ and every two witnesses $w_1, w_2$, $C_L(x^\prime, w_1) = C_L(x^\prime, w_2)$ whenever the flag gates of both are set to $1$. It is clear that this can be done in $\coNP$. For the Arthur-Merlin phase, we run the protocol $\Gamma^M$ of Lemma \ref{lemma:isPromiseProblem}. 

We now show that this protocol correctly decides $A$ given correct advice. First assume that $w = x1^y \in A$, so $x \in L$. Then there is a TSV-circuit $C$ of size $s$, where $(y + |x|)^{k+1} \leq s < (2y + |x|)^{k+1}$. When Merlin gives both verifiers this circuit, the $\coNP$ verifier $V$ will accept. Since $C$ computes $L_{=n}$, by the property of the probabilistic TM $M$ of Lemma \ref{thm:SanthPSPACECompleteLanguage}, $M^{L(C)}(x)$ accepts with probability $1$. Therefore $\langle x, C \rangle \in \Gamma^{M}_Y$, and the protocol accepts. 

Assume that $w = x1^y \notin A$. If $y$ is not of the correct form, then, given the correct advice, the above protocol immediately rejects. If $y$ is of the correct form, then $x \notin L$. Let $C$ be a circuit of size $s$, where $(y + |x|)^{k+1} \leq s < (2y + |x|)^{k+1}$. If $C$ is not PSV, then the $\coNP$ verifier $V$ will reject and the protocol is correct. Otherwise, $C$ is a PSV circuit. By the property of the probabilistic TM $M$ of Lemma \ref{thm:SanthPSPACECompleteLanguage}, $M^{L(C)}(x)$ must reject with probability at least $2/3$. Hence $\langle x, C \rangle \in \Gamma^{M}_N$. Since $C$ was arbitrary, the protocol correctly decides $A$.

We now prove that $A$ does not have SV non-deterministic circuits of size $n^k$. Assume otherwise, and let $C_1, C_2,\ldots$ be a sequence of SV non-deterministic circuits such that $C_m$ decides $A_{=m}$ and $C_m$ is of size $m^k$. Let $s(m)$ be the minimum circuit size of $L_{=m}$. By our assumption, there is an infinite number of input lengths $m$ such that $s(m) > (m+1)^{k+1}$. For any such $m$, define the following circuit $C^\prime_m$ deciding $L_{=m}$. First, the unique value $y$ such that $y$ is a power of $2$ and $(m+y)^{k+1} \leq s(m) < (m+2y)^{k+1}$ is hardcoded into $C^\prime_m$. On input $x$ of length $m$, $C^\prime_m$ simulates $C_{m+y}(x1^y)$. Since the size of $C^\prime_m$ is at most the size of $C_{m + y}$, we have that the size of $C^\prime_m$ is less than $s(m)$. This contradicts our assumption, and the proof is complete.
\end{proof}

We now modify the proof of Theorem \ref{thm:fixedLowerAugAM} slightly to achieve the following theorem.
\begin{theorem}\label{thm:fixedSizeStrong}
For every $k \in \N$ there is a language $A \in (\coaugAM \cap \augAM)/1$ such that $A \notin \nsize(n^k)$.
\end{theorem}
\begin{proof}
The proof is similar to that of Theorem \ref{thm:fixedLowerAugAM}. If $\pspace$ has polynomial size SV circuits, then by Theorem \ref{thm:AvM}, $\pspace = \coaugAM \cap \augAM$, and the claim follows. 

Assume that $\pspace \nsubseteq \svsize(poly)$. Let $k \in \N$, and $L$, $\overline{L}$ be the $\pspace$ complete languages of Lemma \ref{thm:SanthPSPACECompleteLanguage}. For every $n \in \mathbb{N}$, define the $Min(L_n) \in \mathbb{N}$ to be the size of the smallest non-deterministic circuit computing $L_{=n}$. Recall the definition of language $A$,
\begin{center}
$A = \{x1^y \, | \, x \in \overline{L}, \, 0< |x| \leq y, \; y \text{ is a power of 2 and } (y + |x|)^{k+1} \leq Min(L_n) < (2y + |x|)^{k+1}\}.$
\end{center}
We will show that $\overline{A} \in \augAM/1$, where the single bit of advice is the same as the bit to compute $A$. If the advice is set to $0$, then Arthur accepts. If $y$ is not of the correct form, then Arthur accepts. Otherwise, Merlin will send a non-deterministic circuit $C$ to Arthur and the $\coNP$ verifier $V$. If $C$ is not a PSV circuit, then $V$ rejects. Otherwise, Arthur and Merlin run the protocol for $\Gamma^{M^\prime}$ on $\langle x, C \rangle$. Assume that $x1^y \in \overline{A}$. If $y$ is not of the correct form then the above protocol will accept given the correct advice. If $y$ is of the correct form then $x \in \overline{L}$. Therefore, by the property of $\overline{L}$ and $M^\prime$, there is a TSV circuit $C$ such that $\langle x, C \rangle \in \Gamma^{M^\prime}_Y$ and the protocol accepts. Assume that $x1^y \notin \overline{A}$, so $x \notin \overline{L}$. Then for every circuit $C$ Merlin gives to the verifiers, either $C$ is not PSV, and $V$ will reject, or $\langle x, C \rangle \in \Gamma^{M^\prime}_N$, and Arthur will reject. Hence $\overline{A} \in \augAM/1$. Finally, we note that the protocols for $A$ and $\overline{A}$ are given the \textit{same bit} of advice.

The proof that $\overline{A}$ does not have single-valued circuits of size $n^k$ is nearly identical to that of Theorem \ref{thm:fixedLowerAugAM}.

Therefore, for every $k \in \N$, there is a language $A \in (\coaugAM \cap \augAM)/1$ such that $A \notin \svsize(n^k)$. We now extend this to non-deterministic circuits. Assume that for some $c \in \N$, 
\begin{center}
$(\coaugAM \cap \augAM)/1 \subseteq \nsize(n^c)$.
\end{center}
It suffices to show that we can construct a TSV circuit of size $O(n^{c}$ computing any language in $(\coaugAM \cap \augAM)/1$ as follows. Let $A \in (\coaugAM \cap \augAM)/1$, and let $C_A$ and $C_{\overline{A}}$ be size $n^c$ non-deterministic circuits computing $A$ and $\overline{A}$, respectively. Define the circuit $C$ which, given $x \in \sstn$ and $y \in \{0,1\}^{n^c}$, simulates $C_A(x, y)$ and $C_{\overline{A}}(x, y)$. $C$ accepts if $C_A(x, y)$ accepts with its flag bit set, and rejects if $C_{\overline{A}}(x, y)$ accepts with its flag bit set. Then $C$ is a TSV circuit of size $O(n^c)$ computing $A$ and $\overline{A}$, a contradiction.
\end{proof}

Essentially the same proof as Theorem \ref{thm:fixedSizeStrong} shows that $\praugAM$ does not have fixed polynomial size non-deterministic circuits. 
\begin{theorem}\label{thm:fixedSizeStrongPromise}
For every $k \in \N$ there is a language $A \in (\prcoaugAM \cap \praugAM)$ such that $A \notin \nsize(n^k)$.
\end{theorem}
\begin{proof}
The proof is similar to that of Theorem \ref{thm:fixedSizeStrong}, except that we eliminate the one bit of advice through the use of a promise. Recall the language 
\begin{center}
$A = \{x1^y \, | \, x \in L, \, 0< |x| \leq y, \; y \text{ is a power of 2 and } (y + |x|)^{k+1} \leq Min(L_n) < (2y + |x|)^{k+1}\}.$
\end{center}
Our promise consists of all strings $x1^y$ such that $y$ is of the correct form. The remainder of the proof follows the proof of Theorem \ref{thm:fixedSizeStrong}.
\end{proof}

We are now able to prove our downward separation result for non-deterministic circuit size. Note that, with the infinitely often and almost everywhere reversed, the converse is true using standard arguments. That is, if $\SigmaSub$ does not have fixed polynomial size circuits almost everywhere, then $\SigmaE$ does not have polynomial size non-deterministic circuits infinitely often. We will need the following theorem due to Aydinlio\u{g}lu and van Melkebeek \cite{AvM12}.
\begin{theorem}[]\label{thm:AvMequiv}
The following are equivalent.
\begin{enumerate}
\item $\prAM \subseteq \Sigma_2\mathsf{TIME}(2^{n^\epsilon})/n^\epsilon$ for every constant $\epsilon > 0$.
\item $\SigmaE \nsubseteq$ i.o.$-\NPpoly$.
\end{enumerate}
\end{theorem}

The following lemma is implicit in \cite{AvM12}, which we prove for completeness.
\begin{lemma}\label{lemma:amToAugAM}
$\prAM \subseteq \Sigma_2\mathsf{TIME}(2^{n^\epsilon})/n^\epsilon$ for every $\epsilon > 0$ if and only if $\praugAM \subseteq \Sigma_2\mathsf{TIME}(2^{n^\epsilon})/n^\epsilon$ for every $\epsilon > 0$.
\end{lemma}
\begin{proof}
The backward direction is immediate. Let $\Pi = (\Pi_Y, \Pi_N)$ be a promise problem in $\praugAM$. Let $\Gamma = (\Gamma_Y, \Gamma_N)$ be the corresponding $\prAM$ problem for $\Pi$, and $V$ be the $\coNP$ problem for $\Pi$. Let $c \in \N$ be the constant for $\Pi$, and $\epsilon > 0$. By our assumption, there is a $\Sigma_2\mathsf{TIME}(2^{n^\epsilon / c})$ machine $M$ taking $n^{\epsilon / c}$ bits of advice which is consistent with the promise $\Gamma$. Define the $\Sigma_2\mathsf{TIME}(2^{n^\epsilon})$ machine $N$ taking $n^\epsilon$ bits of advice as follows. On input $x \in \sstn$, guess a string $y \in \{0,1\}^{n^c}$, and check if $\langle x, y \rangle \in V$ using the $\NP$ oracle. If it is not, reject. Otherwise, simulate $M$ on $\langle x, y \rangle$ with the given advice string. It is clear that $N$ is a $\Sigma_2\mathsf{TIME}(2^{n^\epsilon})$ time machine taking $n^\epsilon$ bits of advice. Assume that $x \in \Pi_Y$. Then there is a $y \in \{0,1\}^{n^c}$ such that $\langle x, y \rangle \in V$ and $\langle x, y \rangle \in \Gamma_Y$. Therefore, given the correct advice string $\alpha \in \{0,1\}^{n^\epsilon}$, $N$ accepts. Assume that $x \in \Pi_N$. Let $y \in \{0,1\}^{n^c}$ be any string guessed by $N$. If $\langle x, y \rangle \notin V$, then $N$ will reject. Otherwise, $\langle x, y \rangle \in \Pi_N$. Therefore, given the correct advice string $\alpha \in \{0,1\}^{n^\epsilon}$, $N$ rejects. 
\end{proof}

We are now able to prove the downward separation results for non-deterministic circuit size. First, we have the following ``low-end" separation.
\begin{theorem}
If $\SigmaE \nsubseteq$ i.o.-$\NPpoly$, then for every $k \in \N$, there is a language $A \in \SigmaSub$ such that $A \notin \nsize(n^k)$.
\end{theorem}
\begin{proof}
Let $k \in \N$. If $\SigmaEXP \nsubseteq$ i.o.$-\NPpoly$, then by Theorem \ref{thm:AvMequiv} and Lemma \ref{lemma:amToAugAM}, $\praugAM \subseteq \Sigma_2\mathsf{TIME}(2^{n^\epsilon})/n^\epsilon$ for every $\epsilon > 0$. By Theorem \ref{thm:fixedSizeStrongPromise}, there is a language $A \in \praugAM$ such that $A \notin \svsize(n^{2k})$. Let $\epsilon > 0$, and let $M$ be the $\Sigma_2\mathsf{TIME}(2^{n^\epsilon})$ machine deciding $A$ given $n^\epsilon$ bits of advice. We can encode advice into the input as follows. Define the language
\begin{center}
$A^\prime = \{\langle x, \alpha \rangle \, | M \text{ accept } x \text{ given } \alpha \in \{0,1\}^{n^\epsilon} \text{ as advice}\}$.
\end{center}
It is clear that $A^\prime \in \Sigma_2\mathsf{TIME}(2^{n^\epsilon})$. For sufficiently large $n$, $(O(n + n^\epsilon))^k < n^{2k}$. We therefore have that $A^\prime \notin \nsize(n^k)$. As $k$ and $\epsilon$ were chosen arbitrarily, we see that $\SigmaSub \nsubseteq \nsize(n^k)$.
\end{proof}

Theorem \ref{thm:fixedSizeStrong} also implies that derandomizing $\prAM$ in $\SigmaP$ gives fixed polynomial size lower bounds for $\SigmaP$.
\begin{corollary}
If $\prAM \subseteq \SigmaP$, then $\SigmaP \nsubseteq \nsize(n^k)$ for any fixed $k \in \mathbb{N}$.
\end{corollary}
\begin{proof}
Assume that $\prAM \subseteq \SigmaP$. Then $\praugAM \subseteq \SigmaP$. By Theorem \ref{thm:fixedSizeStrongPromise}, $\praugAM \nsubseteq \nsize(n^k)$ for any fixed $k$, and the conclusion follows.
\end{proof}

\section{Mild Derandomization of Promise AM}

\begin{definition}
A $\Sigma_2$-\textit{promise problem} is a promise problem $\Gamma = (\Gamma_Y, \Gamma_N)$ such that there is a language $L \in \SigmaP$ which decides the promise $\Gamma_Y \cup \Gamma_N$. That is, for every length $n$ and all strings $x \in \sstn$,
\begin{center}
$x \in \Gamma_Y \cup \Gamma_N$ if and only if $x \in L$.
\end{center}
The class $\Sigma_2$-$\prAM$ consists of all $\Sigma_2$-promise problems in $\prAM$.
\end{definition}

A polynomial size \textit{hitting set} for a $\Sigma_2$-$\prAM$ problem $\Gamma = (\Gamma_Y, \Gamma_N)$ is a polynomial size set $S$ of $n^k$-bit strings that will take the role of Arthur in the AM protocol. Formally, $S$ is a hitting set if, for every $x \in \Gamma_Y \cup \Gamma_N)$,
\begin{align*}
x \in \Gamma_Y &\implies (\forall y \in S) (\exists z) \, R(x, y, z) = 1 \\
x \in \Gamma_N &\implies (\forall z) (\exists y \in S) \, R(x, y, z) = 0,
\end{align*}
where $R$ is a deterministic polynomial time computable relation for $\Gamma$. Note that we do not worry about the instances which are not in the promise $\Gamma$.

We use Williams easy hitting set technique to give a nontrivial derandomization of $\Sigma_2$-$\prAM$. This is an analog of Williams' result for $\AM$ \cite{Williams16}. We will consider hitting sets for $\Sigma_2$-$\prAM$ which are computable by polynomial size circuits with oracle access to $\NP$. There are two cases. Either there is a constant such that, for every problem $\Gamma$ in $\Sigma_2$-$\prAM$, there is a $n^c$ size hitting set for $\Gamma$ and $\Sigma_2$-$\prAM$ can be computed in $\Poly^{\NP}/O(n^c)$. Otherwise, for every $c$, there is a problem in $\Sigma_2$-$\prAM$ which has no small hitting sets. We can use this fact to find a string of high complexity, and use a pseudorandom generator to derandomize $\Sigma_2$-$\prAM$.
\begin{theorem}
At least one of the following holds.
\begin{enumerate}
\item There is a constant $c \in N$ such that $\Sigma_2$-$\prAM \subseteq \Poly^{\NP}/O(n^c)$.
\item $\prAM \subseteq \text{i.o.-}\Sigma_2\mathsf{TIME}(2^{n^{\epsilon}})/n^\epsilon$ for every $\epsilon > 0$.
\end{enumerate}
In particular, there is a constant $c$ such that 
\begin{center}
$\Sigma_2$-$\prAM \subseteq$ i.o.-$\SigmaSub/n^c$.
\end{center}
\end{theorem}
\begin{proof}
First assume that there exists a constant $c \in \N$ such that, for every $\Sigma_2$-$\prAM$ promise problem $\Gamma$ there is a circuit with oracle access to $SAT$ of size $n^c$ computing a hitting set for $\Gamma$. Then $\Sigma_2$-$\prAM \subseteq \Poly^{\NP}/O(n^c)$. This follows, since the advice is simply the oracle circuit computing the hitting set. 

Otherwise, for every constant $c$, there is a $\Sigma_2$-$\prAM$ problem $\Gamma$ and a polynomial time relation $R$ such that, for infinitely many input lengths, every set hitting set for $\Gamma$ has circuit complexity at least $n^c$. We will show that this fact allows us to compute a string of hard $SAT$-oracle complexity. Once we have such a string, we use the pseudorandom generator of Klivans and van Melkebeek \cite {KvM02} to derandomize $\prAM$. Let $\Pi = (\Pi_Y, \Pi_N)$ be the $\prAM$ promise problem we wish to derandomize. Let $k \in \N$ be the number such that the number of random bits Arthur uses is at most $n^k$. Finally, let $\epsilon > 0$. We show how to compute $\Pi$ in $\Sigma_2\mathsf{TIME}(2^{n^{\epsilon}})/n^\epsilon$. Let $R$ be a polynomial time relation for a $\Sigma_2$-$\prAM$ problem $\Gamma$ such that, for infinitely many input lengths $n$, every hitting set of $\Gamma$ has circuit complexity of at least $n^{2k/\epsilon}$. On input $x \in \sstn$, the $\Sigma_2\mathsf{TIME}(2^{n^{\epsilon}})/n^\epsilon$ first guesses a hitting set $S$ for $\Gamma$ on inputs of size $n^\epsilon$. The advice is the cardinality of $\Gamma_Y$. The machine then guesses three sets of strings, $U_Y, U_N, U_O$ such that $|U_Y| = |\Gamma_Y|$, and $|U_Y| + |U_N| + |U_O| = 2^{n^\epsilon}$. For each string in $U_O$, the machine verifies that it is not in the promise $\Gamma$. By our assumption of the promise, this can be done in $\SigmaP$ time. For each string $x^\prime \in U_Y$, the machine uses its oracle to verify that for every $y \in S$, there is a $z$ such that $R(x^\prime, y, z)$ accepts. Finally, for each string $x^\prime \in U_N$, the machine verifies that there is some $y \in S$ such that every string $z$ satisfies $R(x^\prime, y, z) = 0$. Once the machine has verified each of these items, it then uses the guessed hitting set $S$ and a pseudorandom generator of \cite{KvM02} to derandomize $\Pi$, and accepts if and only if $x \in \Pi_Y$.

Therefore, we have that $\Sigma_2$-$\prAM$ is in either $\Poly^{\NP}/O(n^c)$ for some fixed constant $c$ or $\prAM$ is in $\Sigma_2\mathsf{TIME}(2^{n^{\epsilon}})/n^\epsilon$ for every $\epsilon > 0$, and the claim follows.
\end{proof}

\bibliographystyle{plain}
\bibliography{ComplexityBib}

\end{document}